\documentclass[12pt]{nj_style}

\usepackage{amsmath}
\usepackage{amssymb}
\usepackage{amsthm}
\usepackage{enumerate}
\usepackage[font=small,labelfont=bf,margin=8pt]{caption}
\usepackage[usenames,dvipsnames]{color}
\usepackage{graphicx}
\usepackage[colorlinks=true]{hyperref}
 
\newtheorem{thm}{Theorem} % change to \newtheorem{thm}{Theorem}[section] to number theorems by section

\begin{document}

\begin{frontmatter}


\title{The Complexity of the Puzzles of \emph{Final Fantasy XIII-2}}
\author{Nathaniel Johnston}
\ead{njohns01@uoguelph.ca}
\ead[http://www.njohnston.ca]{www.njohnston.ca}
\address{Department of Mathematics and Statistics, University of Guelph, Guelph, Ontario N1G~2W1, Canada}

\begin{abstract}
	We analyze the computational complexity of solving the three ``temporal rift'' puzzles in the recent popular video game \emph{Final Fantasy XIII-2}. We show that the Tile Trial puzzle is NP-hard and we provide an efficient algorithm for solving the Crystal Bonds puzzle. We also show that slight generalizations of the Crystal Bonds and Hands of Time puzzles are NP-hard.
\end{abstract}


\end{frontmatter}

% \linenumbers % uncomment to add line numbers

%%%%%%%%%%%%%%%%%%%%%%%%%%%%%%%%%%%%%%%%%%%%%%%%%%%%%%%%%%%
\section{Introduction}\label{sec:intro}
%%%%%%%%%%%%%%%%%%%%%%%%%%%%%%%%%%%%%%%%%%%%%%%%%%%%%%%%%%%

The computational complexity of video games is a topic that has earned a significant amount of attention in recent years. Many video games have been proved to be NP-hard, including \emph{Clickomania} \cite{BDD+02}, \emph{Commander Keen} \cite{For10}, \emph{Flood-It!} \cite{ACJ+10}, \emph{Lemmings} \cite{Cor04,For10,Vig12}, \emph{Minesweeper} \cite{Kay00}, \emph{Pac-Man} \cite{Vig12}, \emph{Sokoban} \cite{FG96}, and \emph{Tetris} \cite{DHL03} (see also the survey article \cite{KPS08}).

Most video games that have been studied from a computational complexity point of view are classics: games that were originally developed while video gaming was in its infancy. In the present paper, we instead study the modern video game \emph{Final Fantasy XIII-2}, which was released for the Playstation 3 and XBox 360 video game consoles in early 2012.\footnote{It was originally released in Japan slightly earlier, on December 15, 2011.} In this game, the player is frequently tasked with solving one of three types of ``temporal rift'' puzzles in order to advance. We investigate and derive a hardness result for all three of these puzzles.

The first puzzle, called \emph{Tile Trial}, tasks the player with collecting crystals that are placed on square tiles, while being careful not to step on any tile more than once. We show that the Tile Trial puzzle is NP-hard in Section~\ref{sec:tile_trial}. The second puzzle, called \emph{Crystal Bonds}, is investigated in Section~\ref{sec:crystal_bonds}. This puzzle presents several crystals laid out on a grid, and asks the player to form bonds between some specified pairs of crystals before time runs out. We show that an optimal solution to the Crystal Bonds puzzle can be found in polynomial time, but a natural generalization of the puzzle is NP-hard. Finally, in Section~\ref{sec:hands_of_time} we consider the third puzzle, called \emph{Hands of Time}, in which the player tries to navigate a clock face subject to certain movement restrictions. We show that solving partially-completed Hands of Time puzzles is also an NP-hard problem.

Note that the reader is assumed to be familiar with basic concepts in graph theory. All of our reductions are from various forms of the Hamiltonian path problem -- that is, the problem of finding a path on a graph that visits every vertex exactly once. We use the term \emph{grid graph} to refer to an undirected graph in which every vertex is located at a point on the $2$-dimensional integer lattice $\mathbb{Z}^2$ and there is an edge between two vertices if and only if the (Euclidean) distance between them is $1$.

%%%%%%%%%%%%%%%%%%%%%%%%%%%%%%%%%%%%%%%%%%%%%%%%%%%%%%%%%%%
\section{Puzzle \#1: Tile Trial}\label{sec:tile_trial}
%%%%%%%%%%%%%%%%%%%%%%%%%%%%%%%%%%%%%%%%%%%%%%%%%%%%%%%%%%%

The first puzzle introduced in the game is called \emph{Tile Trial}, which presents the player with a grid of tiles, some of which have crystals on them. The goal is to retrieve all of the crystals and reach the exit, while stepping on each tile no more than once -- see Figure~\ref{fig:tile_trial}. The player's starting tile is always on one far end of the grid, while the finishing tile is on the opposite end of the grid.

\begin{figure}[htb]
	\centering
	\includegraphics[scale=0.7]{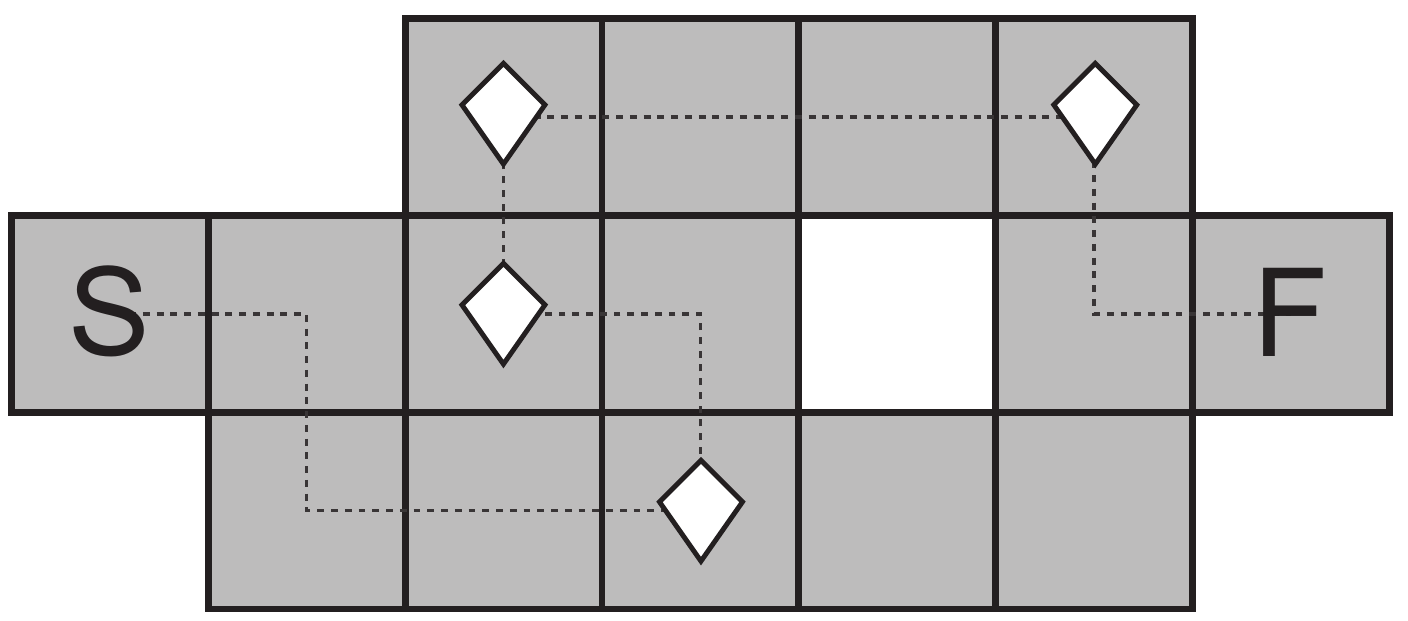}
	\caption{An early-game Tile Trial puzzle. The goal is to move from tile to tile, starting at the tile marked ``S'' and finishing at the tile marked ``F'', in such a way that each tile containing a crystal ($\lozenge$) is touched and no tile is touched more than once. The path traced out by the grey dotted line is a solution to the given puzzle.}\label{fig:tile_trial}
\end{figure}

Some of the more difficult Tile Trial puzzles add in two additional challenges: some tiles may be stepped on twice, and some crystals periodically move between two different tiles. We now show that Tile Trial puzzles are NP-hard even if the crystals never move.

\begin{thm}
	The Tile Trial puzzle is NP-hard (and the associated decision problem is NP-complete).
\end{thm}
\begin{proof}
	Consider an arbitrary grid graph. Construct a Tile Trial puzzle by first placing a tile and a crystal at each of the vertices of the graph. Then choose one of the bottommost tiles and let this tile be stepped on twice. Place two tiles, each of which may be stepped on twice, in a vertical line below this bottommost tile, as in Figure~\ref{fig:tile_trial2}. Finally, create a path from the new bottommost tile that leads left to the player's starting tile, and another path that leads right to the finishing tile.
	
\begin{figure}[htb]
	\centering
	\includegraphics[width=0.95\textwidth]{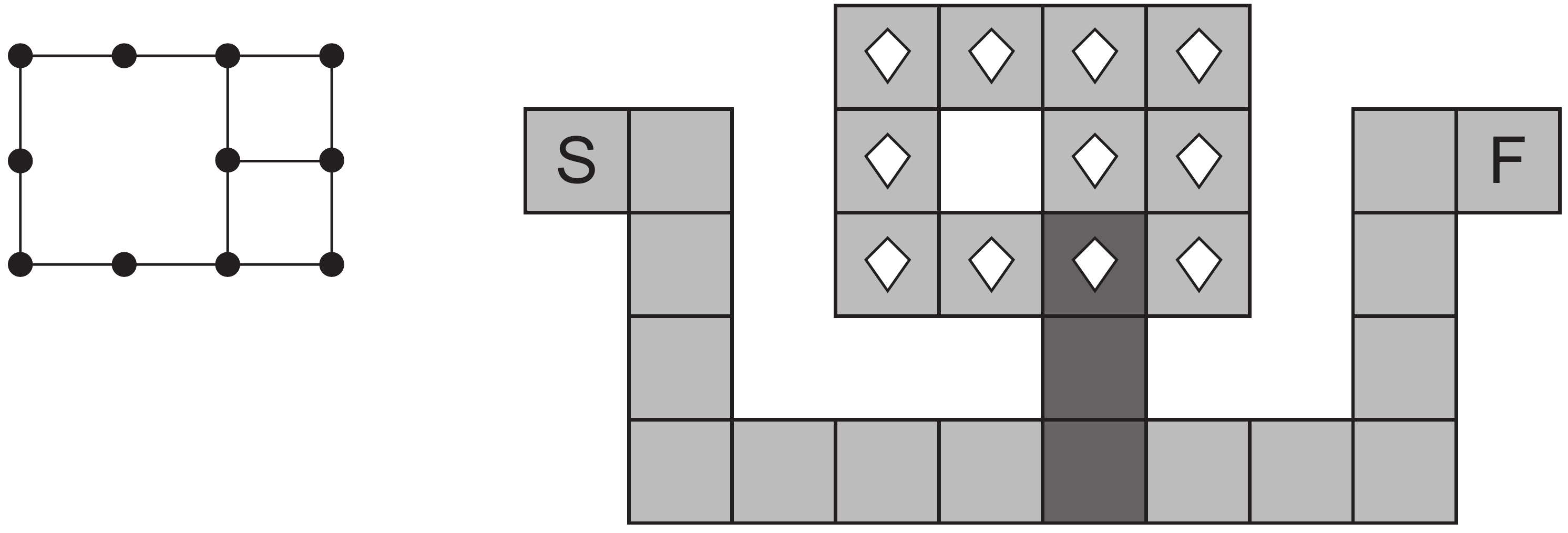}
	\caption{(left) A grid graph. (right) A corresponding Tile Trial puzzle. Dark grey tiles can be stepped on twice. The Tile Trial puzzle is unsolvable, which indicates that the grid graph does not have a Hamiltonian cycle.}\label{fig:tile_trial2}
\end{figure}
	
	It is clear that this Tile Trial puzzle has a solution if and only if the original directed graph has a Hamiltonian cycle. Because the Hamiltonian cycle problem on grid graphs is NP-hard \cite{IPS82}, NP-hardness of Tile Trial puzzles follows.
	
	To see that the problem of deciding whether or not a given Tile Trial puzzle is solvable is NP-complete, note that it is trivial to verify whether or not a given path touches every tile at most once (or twice for certain tiles) and touches every crystal, so this problem is in NP.
\end{proof}

%%%%%%%%%%%%%%%%%%%%%%%%%%%%%%%%%%%%%%%%%%%%%%%%%%%%%%%%%%%
\section{Puzzle \#2: Crystal Bonds}\label{sec:crystal_bonds}
%%%%%%%%%%%%%%%%%%%%%%%%%%%%%%%%%%%%%%%%%%%%%%%%%%%%%%%%%%%

The second puzzle introduced to the player is called \emph{Crystal Bonds}, which presents the player with an orthogonal grid of tiles, some of which contain crystals, and a designated starting tile on a far side of the grid (much like the Tile Trial puzzle). In this puzzle, however, there are straight lines drawn between some crystals, and the goal is to bond all of the indicated crystals by walking between them before time runs out (see Figure~\ref{fig:crystal_bonds}). Because the player walks at a constant speed, we consider the problem of minimizing the distance that needs to be traveled (rather than time) in order to solve the puzzle. Some points of clarification about the rules of this puzzle include:

\begin{figure}[htb]
	\centering
	\includegraphics[scale=0.7]{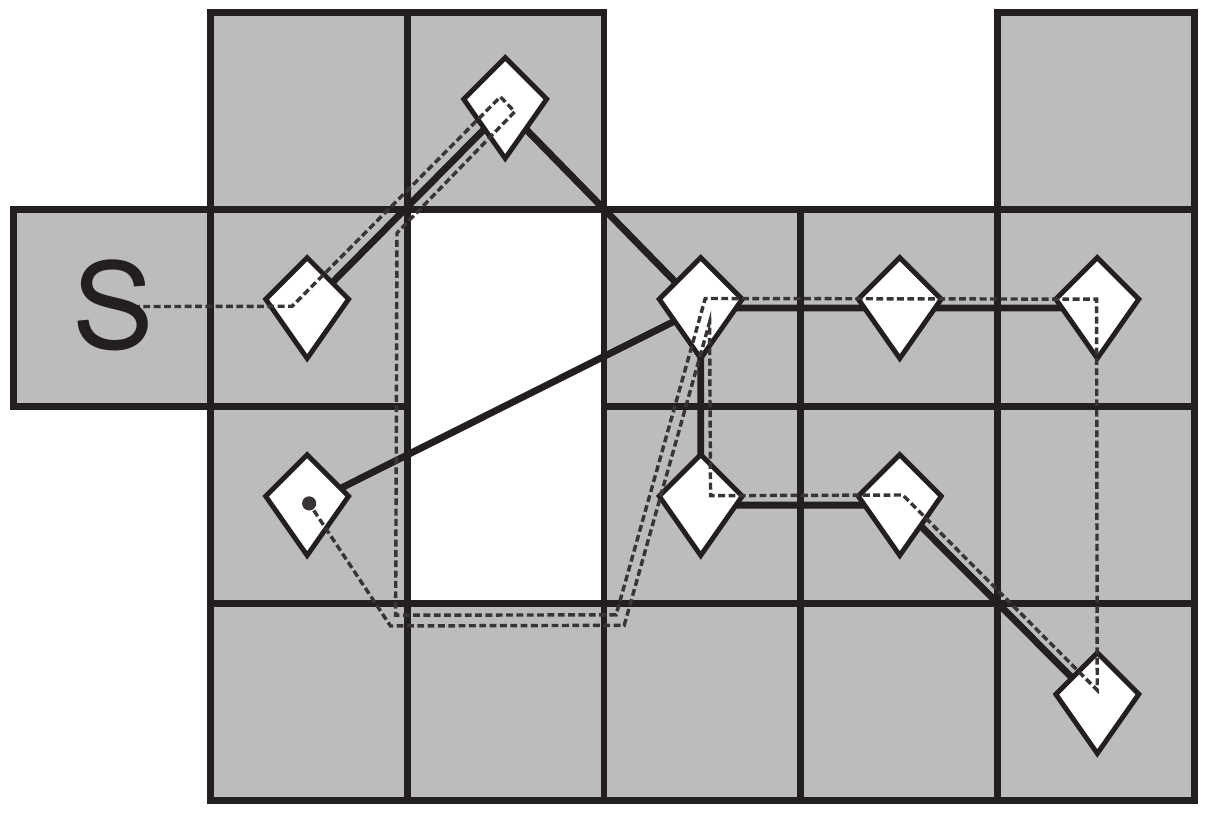}
	\caption{A Crystal Bonds puzzle on a $4 \times 6$ grid. The goal is to bond each pair of crystals ($\lozenge$) that are connected by a black line before time runs out. A solution of minimum distance is overlaid as a grey dotted line. The player starts at the tile marked ``S'', but may end at any crystal of their choosing (the bottom-left crystal in the given solution).}\label{fig:crystal_bonds}
\end{figure}

\begin{itemize}
	\item The player may move in any direction -- they are not restricted to movement orthogonally along the tiles.
	
	\item The player may only move between tiles that are connected orthogonally (i.e., along their edges) -- movement between tiles that are connected only at a corner is not allowed unless there is another tile that connects to both of them orthogonally.
	
	\item Lines between crystals may go over empty space on which the player can not walk. The player's path between those crystals does not need to be a straight line.

	\item The graph of crystal bonds (where crystals are the vertices and the bonds that must be formed are the edges) is always a tree.
	
	\item The player may only create one bond between pairs of crystals at a time. For example, in Figure~\ref{fig:crystal_bonds_link} the player can not create both bonds by walking along the path $A \rightarrow B \rightarrow C$ or $C \rightarrow B \rightarrow A$. However, they can create both bonds by following the longer path $B \rightarrow A \rightarrow C$.
\end{itemize}

\begin{figure}[htb]
	\centering
	\includegraphics[scale=0.7]{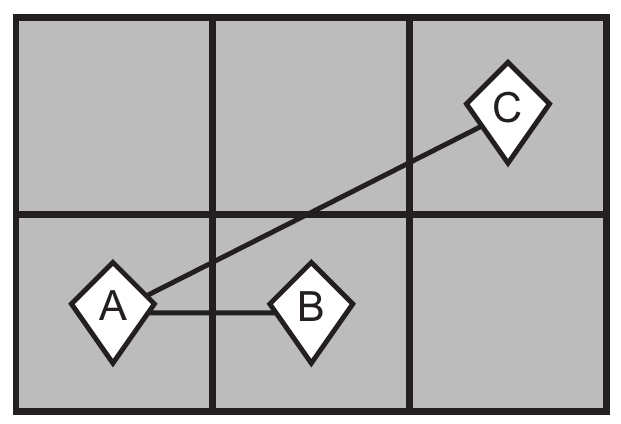}
	\caption{A Crystal Bonds puzzle with $3$ crystals.}\label{fig:crystal_bonds_link}
\end{figure}

Note that in the game, crystals periodically change color, and crystals can only be bonded when they are both the same color. We ignore this additional gameplay element, however, as it does not affect our results. We now show that the optimal solution to the Crystal Bonds puzzle (i.e., the path of minimum length that bonds all of the specified pairs of crystals) can be found efficiently.

%Mathematically, we represent a Crystal Bonds puzzle by a collection $(G_{m,n},S,C)$, where $G_{m,n}$ is a grid graph of height $m$ and width $n$ representing the tiles, with only one vertex $S$ in its leftmost column, representing the player's starting position. $C$ is a tree (representing the bonds between crystals) whose vertices are a subset of $V(G_{m,n}) \backslash S$. Let $T_{m,n}$ be the graph that contains $G_{m,n}$ as a subgraph, but also contains a vertex at the corner of every tile and an edge between two vertices if and only if the player can run between them in a straight line (i.e., if there are no missing tiles between those vertices) -- see Figure~\ref{fig:crystal_bonds_tile_graph}. The goal of the puzzle is then to find the path of minimum length  (measured in the usual Euclidean norm) on $T_{m,n}$ that satisfies the following properties:
%\begin{itemize}
%	\item The path starts at the vertex $S$; and
%		
%	\item The path contains every edge $(v_1,v_2) \in E(C)$ as mutually disjoint subpaths. 
%\end{itemize}
%
%\begin{figure}[htb]
%	\centering
%	\includegraphics[scale=0.7]{crystal_bonds_links}
%	\caption{(left) An arrangement of tiles. (right) The corresponding graph $T_{2,3}$ that contains a vertex at the center and corner of each tile, and edges between vertices that can be walked between in a straight line.}\label{fig:crystal_bonds_tile_graph}
%\end{figure}

\begin{thm}\label{thm:crystal_bonds}
	The optimal solution to a Crystal Bonds puzzle on an $m \times n$ grid with $r$ crystals can be found in $O(r\max\{mn\log(mn),r^2\})$ time.
\end{thm}
\begin{proof}
	We present an algorithm that solves Crystal Bonds puzzles in two steps. The first step can be completed in $O(rmn\log(mn))$ time, while the second step can be completed in $O(r^3)$ time. Taking the longest running time of these two steps gives the result.
	
	Our first step is to compute the minimum walking distance between any two crystals (even between pairs of crystals that we don't need to bond). Given a fixed point in the plane and polygonal obstacles containing a total of $b$ vertices, a \emph{shortest path map} can be computed in $O(b\log(b))$ time via the algorithm of \cite{HS99}. From this shortest path map, the minimum distance between the fixed point and any other point can be computed in $O(\log(b))$ time. The obstacles in this case are the missing tiles, so there are certainly no more than $(m+1)(n+1)$ (i.e., $O(mn)$) vertices of obstacles. Thus we can compute the minimum walking distance between a single crystal and all other crystals in $O(mn\log(mn))$ time, and so we can find the minimum walking distance between all pairs of crystals in $O(rmn\log(mn))$ time.
	
	Our next step is to consider the complete graph $K_r$ whose vertices are the $r$ crystals, and weigh each edge of $K_r$ to be the minimum walking distance between the two crystals at its endpoints. The problem of bonding the specified crystals is now a special case of the \emph{Rural Postman Problem} \cite{Orl74}, which asks for the shortest path on an undirected graph $G$ that uses each edge in some specified subset $E$ of the edges of $G$ at least once. For us, $G = K_r$ and $E$ is the set of edges corresponding to the required bonds between crystals. The Rural Postman Problem can be solved in $O(v^{2c+1}/c!)$ time \cite{Fre79,LR81}, where $v$ is the number of vertices of $G$ and $c$ is the number of connected components of the subgraph of $G$ whose edges are the elements of $E$. In our situation, we have $c = 1$ because the set $E$ forms a tree and is thus connected, and $v = r$. It follows that this step can be completed in $O(r^3)$ time.
\end{proof}

Even though Crystal Bonds puzzles can be solved in polynomial time, there is a natural generalization that is $NP$-hard. We call a puzzle that is the same as a Crystal bonds puzzle, except which allows the graph of crystal bonds to be disconnected (rather than forcing it to be a tree), a Disconnected Crystal Bonds puzzle. Recall that the algorithm of Theorem~\ref{thm:crystal_bonds} relied on finding a solution to the Rural Postman Problem. Even though the Rural Postman Problem can be solved in polynomial time as long as the set of required edges is connected (or even if it contains a constant number of connected components \cite{Fre79,Orl76}), it becomes NP-hard if the number of connected components is allowed to grow \cite{LR76}.

The Crystal Bonds puzzle behaves similarly. We saw that it is solvable in polynomial time when the graph of crystal bonds is connected, and a similar argument works to show that it is solvable in polynomial time if we let the graph of crystal bonds have a constant number of connected components. However, we now show that if we place no restrictions on the crystal bonds, then the puzzle becomes NP-hard.

\begin{thm}
	The Disconnected Crystal Bonds puzzle is NP-hard (and the associated decision problem is NP-complete).
\end{thm}
\begin{proof}
	A proposed solution of the Disconnected Crystal Bonds puzzle can have its length verified in polynomial time, so the decision problem version of the puzzle (i.e., the problem that asks whether there exists a solution with walking distance $\leq k$) is in NP. We now demonstrate NP-hardness.
	
	We prove the result via a reduction from the NP-hard problem of finding Hamiltonian paths on grid graphs \cite{IPS82}. Given a grid graph with $v$ vertices, we construct a Disconnected Crystal Bonds puzzle by placing a tile with a crystal on it at each vertex of the grid graph and $2v$ tiles along each edge of the grid graph, as in Figure~\ref{fig:disconnected_crystal_bonds}. We also place one crystal in a single (arbitrarily-chosen) tile adjacent to each of the tiles already containing a crystal, and we draw lines connecting these adjacent pairs of crystals.

\begin{figure}[htb]
	\centering
	\includegraphics[width=0.95\textwidth]{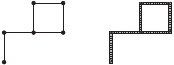}
	\caption{(left) A grid graph with $v = 6$ vertices. (right) A corresponding Disconnected Crystal Bonds puzzle with minimal solution length $65 \leq (v-1)(2v+1) + 2v = 77$, indicating that the graph on the left has a Hamiltonian path.}\label{fig:disconnected_crystal_bonds}
\end{figure}

Walking along an edge of the grid graph corresponds to the player walking a distance of $2v+1$ tiles in a straight line in the Disconnected Crystal Bonds puzzle. Once the player is at a tile corresponding to a vertex, they can connect the crystal there to the adjacent crystal by walking a distance of no more than $2$: one tile to the adjacent crystal, and one tile back. Thus, if the original grid graph has a Hamiltonian path, then the Disconnected Crystal Bonds puzzle has a solution of length no longer than $(v-1)(2v+1) + 2v$. On the other hand, if the original grid graph does not have a Hamiltonian path, then the optimal solution to the Disconnected Crystal Bonds puzzle has length at least $v(2v+1) > (v-1)(2v+1) + 2v$, and NP-hardness follows.

For simplicity, we ignored the fact that the starting tile is always on a far side of the grid and we constructed a Disconnected Crystal Bonds puzzle as if we could start wherever we like. This technicality can be accounted for as follows. If one of the leftmost vertices of the grid graph has degree $1$, simply place the starting tile adjacent to the corresponding crystal in the Disconnected Crystal Bonds puzzle. Otherwise, choose one of the leftmost vertices of the grid graph with degree $2$ and place the starting tile adjacent to the corresponding crystal. Furthermore, alter the puzzle by removing the empty tile that is adjacent to this crystal, extending the newly-broken path by a length of two, and adding a pair of crystals on the two new tiles, as in Figure~\ref{fig:disconnected_crystal_bonds2}. Now a solution of length $\leq v(2v+1) + 2v$ indicates a Hamiltonian cycle (rather than a Hamiltonian path), while the optimal solution having length $> v(2v+1) + 2v$ indicates that the graph has no Hamiltonian cycle.

\begin{figure}[htb]
	\centering
	\includegraphics[scale=1.2]{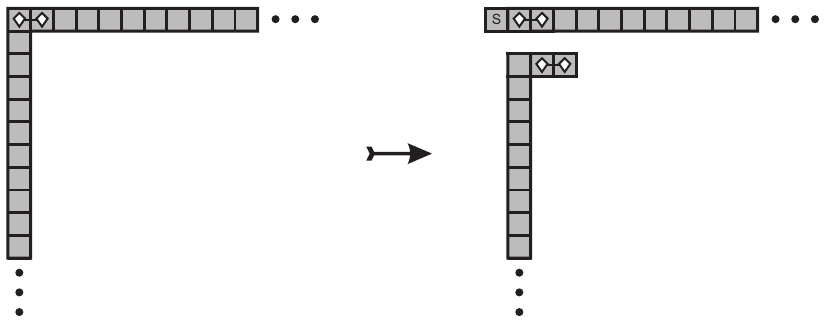}
	\caption{An illustration of the procedure to insert a starting tile while retaining NP-hardness.}\label{fig:disconnected_crystal_bonds2}
\end{figure}
\end{proof}

%%%%%%%%%%%%%%%%%%%%%%%%%%%%%%%%%%%%%%%%%%%%%%%%%%%%%%%%%%%
\section{Puzzle \#3: Hands of Time}\label{sec:hands_of_time}
%%%%%%%%%%%%%%%%%%%%%%%%%%%%%%%%%%%%%%%%%%%%%%%%%%%%%%%%%%%

The third and final puzzle to be introduced in \emph{Final Fantasy XIII-2} is called \emph{Hands of Time}. This puzzle presents the player with $n$ nodes arranged around a circular clock face, and on each node is a positive integer from $1$ to $\lfloor n/2 \rfloor$ (inclusive). The rules of the puzzle are as follows:
\begin{enumerate}
	\item The player starts by selecting one of the $n$ nodes on the clock face. Let's call the number in the selected node $m$. Upon selecting this node, its number vanishes, leaving the node empty.
	
	\item The player now has the option of selecting either the node $m$ positions clockwise from their last choice, or $m$ positions counter-clockwise from their last choice. The value $m$ is updated to be the number on the newly-selected node. That number then vanishes from the node, leaving the node empty.

	\item The player may never select an empty node (i.e., a node they have already selected in the past). Step 2 is repeated until they have no valid moves remaining.
\end{enumerate}

The goal of the puzzle is to select every node on the clock face before running out of valid moves.

In order to analyze Hands of Time puzzles, we first note that they can be represented as directed graphs in which the vertices represent the $n$ nodes on the clock face and the directed edges represent the valid moves clockwise and counter-clockwise from each node. A Hands of Time puzzle then tasks the player with finding a directed Hamiltonian path on this directed graph -- see Figure~\ref{fig:hands_of_time}.

\begin{figure}[htb]
	\centering
	\includegraphics[scale=0.3]{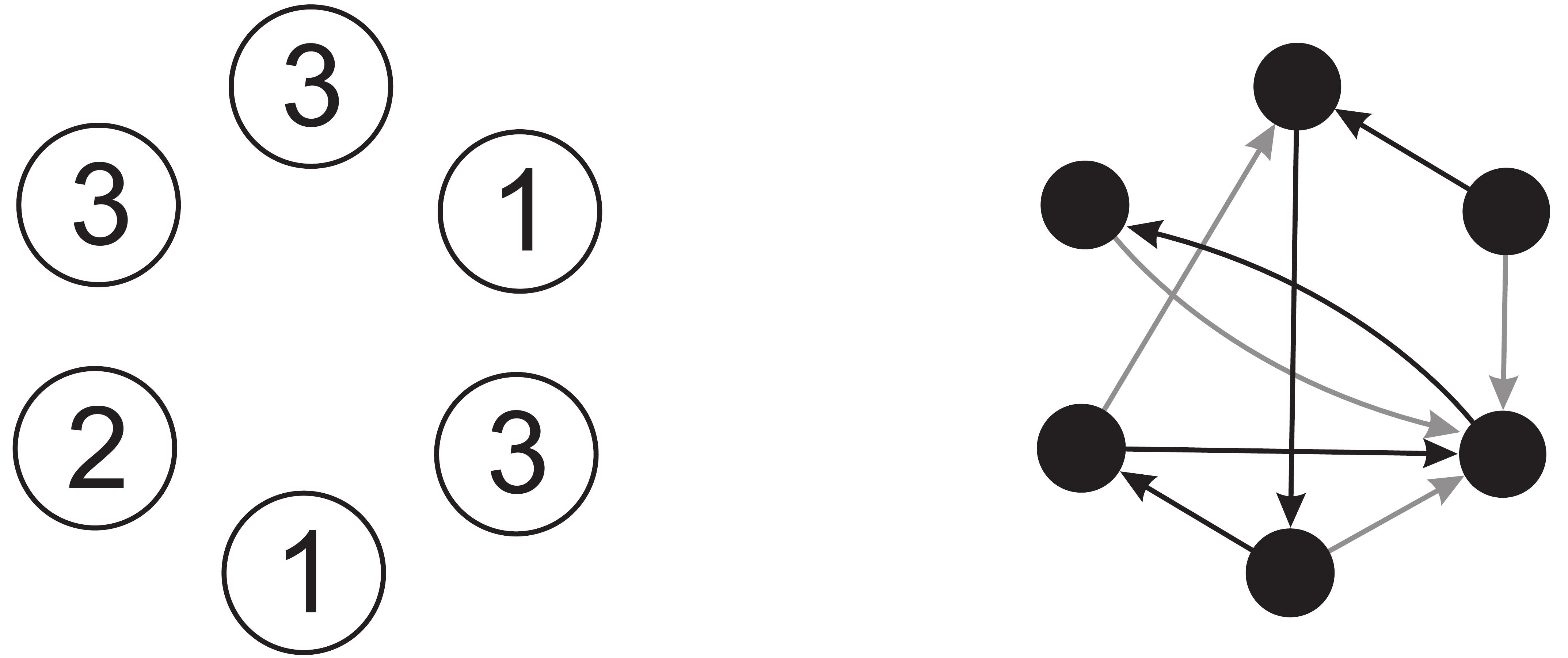}
	\caption{(left) A Hands of Time puzzle as it appears in-game, with $n = 6$ nodes. The numbers on the nodes indicate how many positions clockwise or counter-clockwise the player is allowed to jump. (right) The same Hands of Time puzzle, depicted as a directed graph. The black directed edges indicate a solution to the puzzle (i.e., a directed Hamiltonian path on the graph), while the grey directed edges indicate valid moves that are not part of that solution.}\label{fig:hands_of_time}
\end{figure}

Finding directed Hamiltonian paths in general is NP-hard \cite{Kar72}, even when each vertex has outdegree $1$ or $2$ \cite{Ple79}, so we might expect that the Hands of Time puzzle is NP-hard. On the other hand, if we allow the player to move at most $m$ nodes clockwise or counter-clockwise in step 2 of the puzzle (rather than exactly $m$ nodes), the problem becomes a directed analogue of the problem of finding a Hamiltonian path on a circular-arc graph, which can be solved in polynomial time \cite{Dam93,HCL09}. The complexity of the Hands of Time puzzle is thus not clear.

We now show that a slight generalization of the Hands of Time puzzle is NP-hard: the problem of solving a Hands of Time puzzle that has potentially been partially solved. We call a puzzle that plays the same as a Hands of Time puzzle, but potentially has some empty nodes around the clock face, a \emph{Partial Hands of Time} puzzle -- see Figure~\ref{fig:partial_hands_of_time}. Note that a Partial Hands of Time puzzle can be encoded by listing the positions of non-empty nodes and the numbers on them, and we assume this encoding (or one of comparable efficiency) throughout the following proof. Note in particular that, under this encoding, the number of empty nodes around the clock face can be exponential in $n$ (the number of non-empty nodes) while keeping the input length polynomial in $n$.

\begin{figure}[htb]
	\centering
	\includegraphics[scale=0.3]{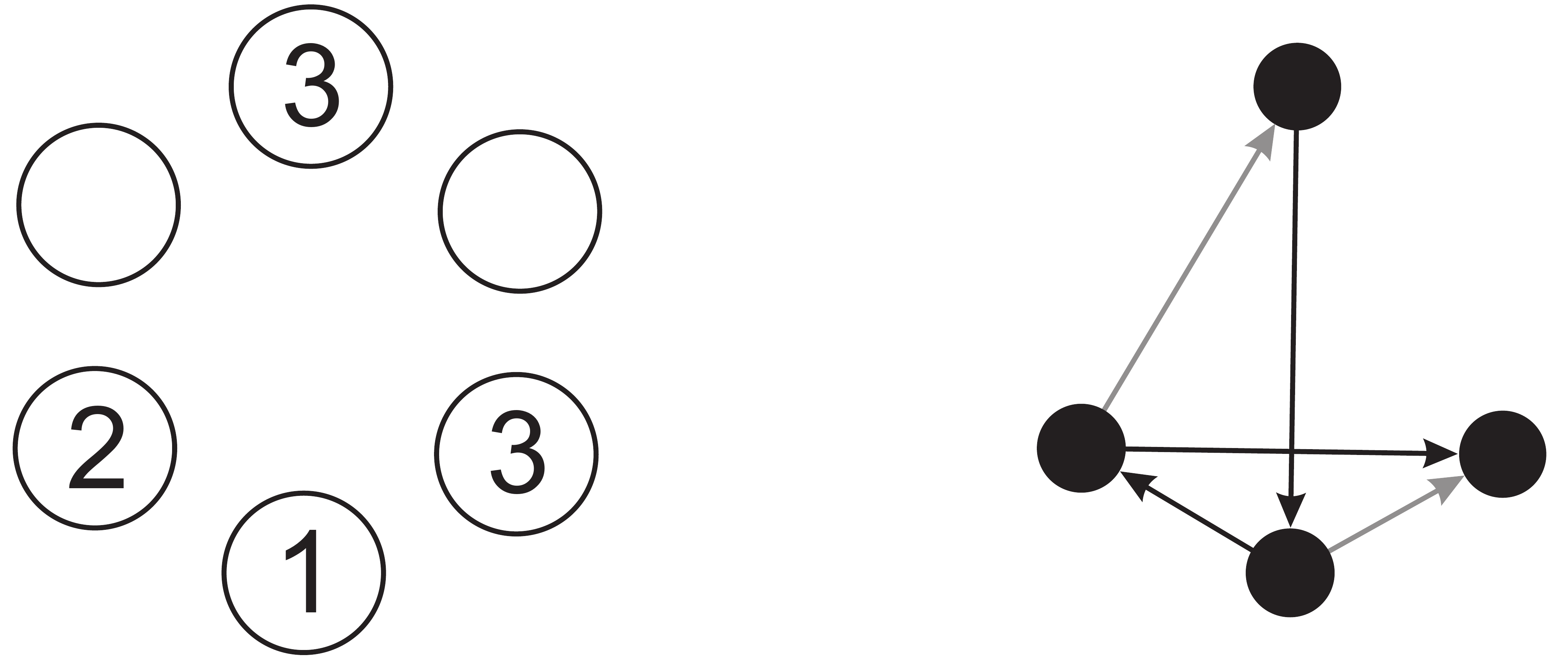}
	\caption{(left) A Partial Hands of Time puzzle -- the same as a Hands of Time puzzle, but with some empty nodes. (right) The directed graph representation of the same Partial Hands of Time puzzle.}\label{fig:partial_hands_of_time}
\end{figure}

\begin{thm}
	The Partial Hands of Time puzzle is NP-hard (and the associated decision problem is NP-complete).
\end{thm}
\begin{proof}
	As with the previous puzzles, it is trivial to verify a proposed solution to a Partial Hands of Time puzzle, so we only need to show that it is NP-hard.
	
	We prove NP-hardness via a reduction from the problem of finding a Hamiltonian path on a directed graph in which each vertex has outdegree $1$ or $2$, which is known to be NP-hard \cite{Ple79}. We construct a Partial Hands of Time puzzle whose associated directed graph (as in Figure~\ref{fig:partial_hands_of_time}) is a given arbitrary directed graph with vertices of outdegree $1$ or $2$, plus perhaps some extra (irrelevant) vertices that have indegree and outdegree both equal to $1$ along some edges. Because the Partial Hands of Time puzzle has a solution if and only if the associated directed graph has a Hamiltonian path, and vertices with indegree and outdegree equal to $1$ don't affect the presence of Hamiltonian paths, NP-hardness follows.
	
	To this end, fix a directed graph with $v$ vertices, and consider a Partial Hands of Time puzzle with $v$ nodes (we refer to these original $v$ nodes as \emph{primary nodes}). Label the primary nodes (reading clockwise, choosing a starting point arbitrarily) $(0,0), (1,0), (2,0), \ldots, (v-1,0)$ and associate the primary nodes bijectively (and arbitrarily) with the $v$ vertices of the directed graph. Insert $10^k - 1$ nodes between nodes $(k,0)$ and $(k+1,0)$ for $0 \leq k \leq v-2$ and insert $10^{v-1} - 1$ nodes between nodes $(v-1,0)$ and $(0,0)$ (we refer to these newly-added nodes as \emph{secondary nodes}). Label the secondary nodes inserted between $(k,0)$ and $(k+1,0)$ (again, reading clockwise) as $(k,1), (k,2), \ldots, (k,10^k - 1)$ and label the secondary nodes inserted between $(v-1,0)$ and $(0,0)$ as $(v-1,1), (v-1,2), \ldots, (v-1,10^{v-1} - 1)$ -- see Figure~\ref{fig:partial_hands_of_time_construct}.

\begin{figure}[htb]
	\centering
	\includegraphics[width=0.95\textwidth]{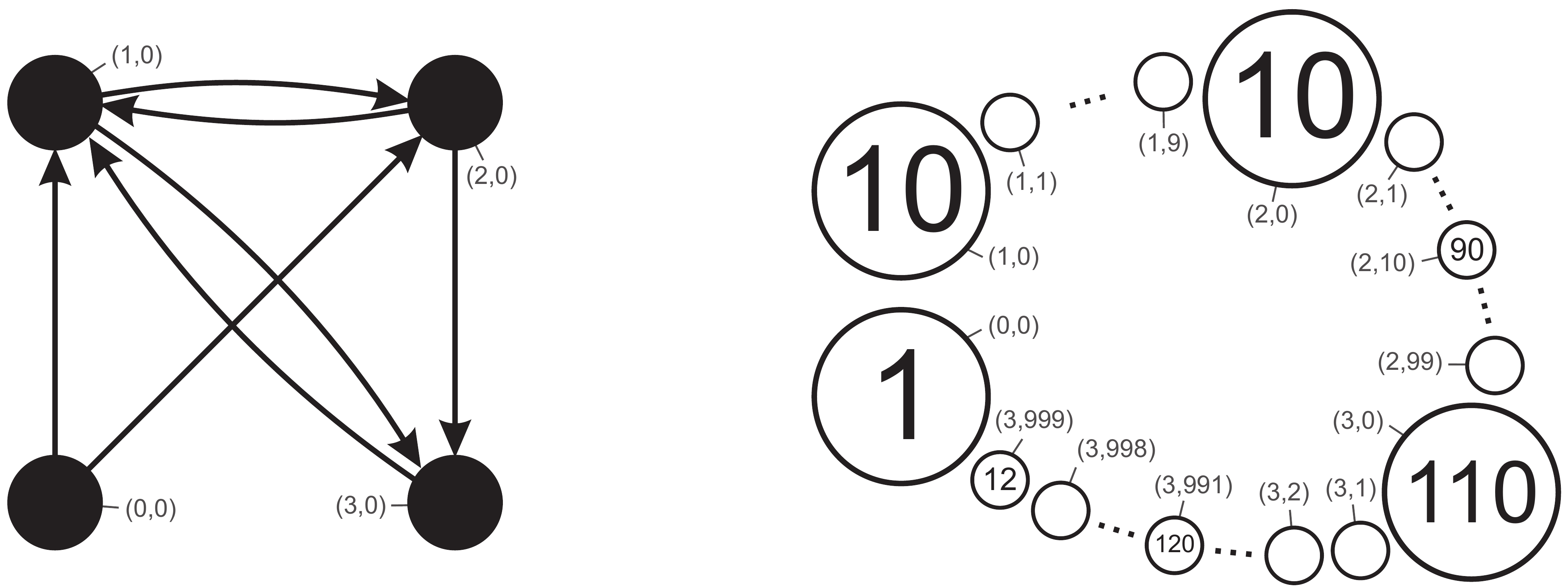}
	\caption{(left) A directed graph in which each vertex has outdegree $1$ or $2$. (right) A Partial Hands of Time puzzle constructed from the graph on the left. The four large circles represent the primary nodes (which correspond to the four vertices in the directed graph), while the small circles represent the secondary nodes. The labels for the nodes are written in grey, while the black numbers within the circles, represent how far the player moves from that node.}\label{fig:partial_hands_of_time_construct}
\end{figure}

% \newline Each edge in the directed graph corresponds to either a single jump on the Partial Hands of Time puzzle, or a sequence of two jumps. For example, the edge from the bottom-left vertex to the top-left vertex corresponds to the jump $(4,0) \rightarrow (1,0)$, while the edge from the bottom-left vertex to the top-right vertex corresponds to the sequence of jumps $(4,0) \rightarrow (3,999) \rightarrow (2,0)$.
	
	We now place integers on the nodes of the Partial Hands of Time puzzle so that it has a solution if and only if the given directed graph has a Hamiltonian path. Each vertex $(j,0)$ in the graph has a directed edge leading to a vertex $(k,0)$, and possibly a second directed edge leading to another vertex $(m,0)$, and we choose this labeling so that $k < m$. The procedure for placing numbers on the nodes is as follows:
	\begin{enumerate}
		\item On node $(j,0)$, place the integer $d_{jk} := \sum_{i=\min\{j,k\}}^{\max\{j,k\}-1}10^i$, which allows the player to move from node $(j,0)$ to node $(k,0)$.
		\item By placing an integer on the node $(j,0)$ in the previous step, the player may now also move the distance $d_{jk}$ away from node $(k,0)$. If the outdegree of vertex $(j,0)$ is $2$, we now consider three cases to make use of this fact:
		\begin{enumerate}
			\item If $k < m < j$, place on node $(j,d_{jk})$ the integer $d_{jkm}$, which we define as $d_{jkm} := d_{jm} + d_{jk}$.
			\item If $k < j < m$, place on node $(j,d_{jk})$ the integer $d_{jkm}$, which we define as $d_{jkm} := d_{jm} - d_{jk}$.
			\item If $j < k < m$, place on node $(v-1,10^{v-1} + d_{0j} - d_{jk})$ the integer $d_{jkm}$, which we define as $d_{jkm} := d_{jm} + d_{jk}$.
		\end{enumerate}
		Note that in all three cases we have placed a number on a secondary node that makes the sequence of moves $(j,0) \rightarrow (\text{secondary node}) \rightarrow (m,0)$ possible. Note also that the puzzle has at most $2v$ non-empty nodes ($v$ primary and $\leq v$ secondary).
	\end{enumerate}

	We now prove that this procedure produces a Partial Hands of Time puzzle that is solvable if and only if the original directed graph has a Hamiltonian path. For the most part this is clear -- we constructed the puzzle in such a way that each directed edge of the graph corresponds to either one or two moves in the puzzle. However, we must show that we have not introduced any new paths through the puzzle -- that is, we have to show that there is no way to land on a non-empty node via a move other than one of the moves suggested above.
	
	When moving from a primary node, this is clear because it is easily verified that
	\begin{itemize}
		\item $d_{jk} \neq d_{j^\prime k^\prime}$ when $\{j,k\} \neq \{j^\prime,k^\prime\}$, and
		\item the first decimal digit of $10^{v-1} + d_{0j} - d_{jk}$ in case 2(c) above is $8$ or $9$, while $d_{jk}$ in cases 2(a) and 2(b) contains only the digits $0$ and $1$ (so $10^{v-1} + d_{0j} - d_{jk} \neq d_{j^\prime k^\prime}$ for all $(j,k) \neq (j^\prime,k^\prime)$).
	\end{itemize}
	
	Before considering the case where the player moves from a non-empty secondary node, we first note from step 2 above that each non-empty secondary node is of the form $(\cdot,r)$, where $r$ is a number whose decimal expansion contains either only $0$ and $1$ (cases 2(a) and 2(b)) or only $1$, $8$, and $9$ (case 2(c)). Similarly, $d_{jkm}$ is an integer whose decimal expansion contains either only $0$, $1$, and $2$ (cases 2(a) and 2(c)) or only $0$, $1$, $8$, and $9$ (case 2(b)).
	
	If we start on a non-empty secondary node and move away from $(m,0)$ in case 2(a) above, we end up at the secondary node $(j,d_{jk} + d_{jkm})$. Since $d_{jk} + d_{jkm}$ contains a $2$ as its $j$th decimal digit, this secondary node must be empty. If we move away from $(m,0)$ in case 2(b) above, we end up at the secondary node $(v-1,10^{v-1} + d_{0j} + d_{jk} - d_{jkm})$. Since $10^{v-1} + d_{0j} + d_{jk} - d_{jkm}$ contains a $2$ as its $j$th decimal digit, this secondary node must be empty. Finally, if we move away from $(m,0)$ in case 2(c) above, we end up at the secondary node $(v-1,10^{v-1} + d_{0j} - d_{jk} - d_{jkm})$. The distance between this node and a non-empty secondary node from case 2(c) above is $|(10^{v-1} + d_{0j} - d_{jk} - d_{jkm}) - (10^{v-1} + d_{0j^\prime} - d_{j^\prime k^\prime})| = |d_{j^\prime j} - 3d_{jk} - d_{km} + d_{j^\prime k^\prime}|$, if we assume without loss of generality that $j \geq j^\prime$. Since the $3d_{jk}$ term ensures that this last quantity is non-zero for any values of $j,k,m,j^\prime,k^\prime$, we see that the secondary node $(v-1,10^{v-1} + d_{0j} - d_{jk} - d_{jkm})$ is empty, completing the proof.
\end{proof}

%%%%%%%%%%%%%%%%%%%%%%%%%%%%%%%%%%%%%%%%%%%%%%%%%%%%%%%%%%%
\section{Discussion}\label{sec:discussion}
%%%%%%%%%%%%%%%%%%%%%%%%%%%%%%%%%%%%%%%%%%%%%%%%%%%%%%%%%%%

While we have proved NP-hardness for the Tile Trial puzzle and shown that the Crystal Bonds puzzle can be efficiently solved, our only result on the Hands of Time puzzle concerns its partially-solved generalization. A strengthening of our result on the Hands of Time puzzle would be of interest. In particular, our proof of NP-hardness depends on the construction of a Partial Hands of Time puzzle in which there are a potentially exponential number of empty nodes; does the puzzle remain NP-hard if the number of empty nodes is promised to be polynomial in the number of non-empty nodes (i.e., is the Partial Hands of Time puzzle strongly NP-hard)? Better yet, is the Hands of Time puzzle itself NP-hard?

\vspace{0.1in}\noindent{\bf Acknowledgements.} The author was supported by the University of Guelph Brock Scholarship.

\bibliographystyle{alpha}

\end{document}